\documentclass[11pt]{article}

\usepackage{hyperref}       % hyperlinks
\usepackage{booktabs}       % professional-quality tables
\usepackage{nicefrac}       % compact symbols for 1/2, etc.
\usepackage{microtype}      % microtypography
\usepackage{framed}

\usepackage{bm}
\usepackage{enumerate}
\usepackage{wrapfig}

\usepackage{amsmath,amssymb,amsthm}
\usepackage{graphicx,color,xcolor}
\usepackage{dsfont}
\usepackage{booktabs}
\usepackage[sort,numbers]{natbib}
\usepackage[sort,nocompress,noadjust]{cite}
\usepackage[normalem]{ulem}
\usepackage[boxed]{algorithm}
\usepackage{algpseudocode}
\usepackage{mathtools}
\usepackage[nameinlink,capitalize]{cleveref}

\newtheorem{theorem}{Theorem}

\newtheorem{fact}{Fact}
\newtheorem{lemma}{Lemma}

\theoremstyle{definition}
\newtheorem{defin}{Definition}

\newtheorem*{question}{Question}
\newtheorem{problem}{Problem}

\newcommand{\cB}{\mathcal{B}}

\newcommand{\cM}{\mathcal{M}}

\newcommand{\cX}{\mathcal{X}}

\newcommand{\RR}{\mathbb{R}}

\newcommand{\bone}{\mathbf{1}}

\newcommand*{\ep}{\varepsilon}
\newcommand*{\eps}{\varepsilon}
\newcommand*{\defeq}{:=}
\newcommand*{\poly}{\mathrm{poly}}
\newcommand*{\dd}{\, \mathrm{d}}

\DeclareMathOperator*{\argmin}{argmin}

\newcommand*{\Phat}{\hat{P}}

\newcommand*{\What}{\hat{W}}

\newcommand*{\Ktilde}{\tilde{K}}

\newcommand*{\Coup}{\cM(\p, \q)}
\newcommand*{\Coupp}{\cM(\p', \q')}

\newcommand*{\p}{\mathbf{p}}
\newcommand*{\q}{\mathbf{q}}

\newcommand*{\Otilde}{\tilde{O}}

\DeclareMathOperator{\half}{\frac{1}{2}}

\newcommand{\R}{\mathbb R}

\newcommand*{\TaylorGKM}{\textsc{TaylorGKM}}

\newcommand{\ball}{\cB_2^d}
\newcommand{\diag}{\mathbb{D}}

\newcommand{\errr}{\text{err}_r}
\newcommand{\errc}{\text{err}_c}

\newcommand*\samethanks[1][\value{footnote}]{\footnotemark[#1]}

\usepackage[margin=1.25in]{geometry} 

\title{Approximating the Quadratic Transportation Metric\\
 in Near-Linear Time}

\author{
        Jason Altschuler\footnote{Supported in part by NSF Graduate Research Fellowship 1122374.}\\
        MIT\\
        \texttt{jasonalt@mit.edu}
        \and
        Francis Bach\footnote{Supported in part by the European Research Council (grant SEQUOIA 724063).} \\
        INRIA \\
        \texttt{francis.bach@inria.fr}
        \and
        Alessandro Rudi\samethanks\\
        INRIA \\
        \texttt{alessandro.rudi@inria.fr}
        \and 
        Jonathan Weed\footnote{Supported in part by the Josephine de K\'arm\'an Fellowship.} \\
        MIT \\
        \texttt{jweed@mit.edu}
}

\begin{document}

\date{}
\maketitle

% !TEX root = ../stoc.tex
\begin{abstract}
Computing the quadratic transportation metric (also called the $2$-Wasserstein distance or root mean square distance) between two point clouds, or, more generally, two discrete distributions, is a fundamental problem in machine learning, statistics, computer graphics, and theoretical computer science. A long line of work has culminated in a sophisticated geometric algorithm due to~\citet{AgaSha14}, which runs in time $\tilde{O}(n^{3/2})$, where $n$ is the number of points. However, obtaining faster algorithms has proven difficult since the $2$-Wasserstein distance is known to have poor sketching and embedding properties, which limits the effectiveness of geometric approaches.
In this paper, we give an extremely simple deterministic algorithm with $\tilde{O}(n)$ runtime by using a completely different approach based on entropic regularization, approximate Sinkhorn scaling, and low-rank approximations of Gaussian kernel matrices. We give explicit dependence of our algorithm on the dimension and precision of the approximation.
\end{abstract}

%\tableofcontents

% !TEX root = ../stoc.tex
\section{Introduction}\label{sec:intro}
Transportation metrics---known in various communities as Wasserstein distances, Kantorovich distances, and optimal transport distances---are a natural set of metrics on probability distributions supported on metric spaces with widespread applications throughout mathematics and statistics~\citep{Vil08}.

Given two distributions $\p$ and $\q$ on $\RR^d$, we define
\begin{equation}\label{eq:wp}
W_p(\p, \q) \defeq \min_{\gamma \in \Coup} \left(\int \|x - y\|^p \dd \gamma(x, y)\right)^{1/p}\,,
\end{equation}
where $\|\cdot\|$ is the Euclidean distance on $\RR^d$ and where $\Coup$ is the set of couplings between $\p$ and $\q$, that is, the set of all joint distributions on $\RR^d \times \RR^d$ whose projections onto the first and second component agree with $\p$ and $\q$, respectively (see, e.g.,~\citep{Vil03, San15} for background on Wasserstein distances and their properties).
Wasserstein distances have recently witnessed a surge of interest from the computer science community in large part due to their effectiveness in a number of practical domains, including image processing and retrieval~\citep{RubTomGui00,PelWer08,LazSchPon06} and machine learning~\citep{JitSzaChw16,ArjChiBot17,BouGelTol17}.

Of particular interest to applications is the $2$-Wasserstein distance $W_2$, which has been used for barycenter computation~\citep{BonRab15, CutDou14}, shape interpolation~\citep{BonVan11,SolGoePey15}, shape reconstruction~\citep{DeCoh11}, triangulations~\citep{MullMem11}, domain adaptation~\citep{CouFla15}, synthesis of soft maps~\citep{SolGuiBut13}, blue-noise generation~\citep{DeBree12}, and many more. For many such applications, $W_2$ gives better practical results than other $W_p$ distances (in particular $W_1$), see e.g. the discussions in~\citep{CouFla15,AndNaoNei16,MullMem11,DeBree12} and the references within.

There has been a great deal of work on the question of how fast $W_p$ can be computed between discrete distributions. Much of the focus, for computational reasons, has been on~$W_1$, also known as earth mover's distance. This research direction, inaugurated by embedding results into the $\ell_1$ metric~\citep{Cha02,Ind03}, has focused largely on efficiently \emph{sketching} $W_1$ (see, e.g.,~\citep{Ind04}), and has resulted in an algorithm to compute a $(1+\ep)$ multiplicative approximation to $W_1$ between two distributions supported on $n$ atoms in time $\Otilde(n)$, that is, \emph{nearly linear} in the size of the input~\citep{ShaAga12} (see also~\citep{Ind07,AgaSha14,AndNikOna14}).\footnote{In this section, we suppress dependence on the dimension $d$ and precision $\ep$ and defer detailed consideration to Section~\ref{sec:intro:rel}.}
These impressive results rely strongly on the fact that the cost $\|x - y\|$ appearing in the definition of $W_1$ is a metric.% \note{FB: here I would detail what is hidden in the notation $\Otildede$: defining it but also digging into the results above to get this dependence explicitly, potentially as a footnote}

%This distance, also known as the \emph{earth mover distance}, can be computed efficiently over $\RR^d$ due to the presence of efficient \emph{sketches} of $W_1$~\citep{}

In contrast to these results for $W_1$, the situation for $W_2$ -- also known as the \emph{root mean square} (RMS) distance or the transportation metric with quadratic cost -- is much less complete. Breakthroughs due to \citet{PhiAga06} (for the $\RR^2$ case) and \citet{AgaSha14} (for the general case) showed that this quantity can be approximated in time $\Otilde(n^{3/2})$, but no better results are known. This lack of progress is partially explained by the fact that the cost $\|x - y\|^2$ appearing in the definition of~$W_2$ is not a metric.
Moreover, strong evidence was given for the difficulty of approximating $W_2$ by~\citet{AndNaoNei16}, who showed that, unlike the earth mover distance, the quadratic transportation metric cannot be effectively sketched. This impossibility result poses a fundamental obstacle to geometric algorithms for estimating $W_2$ in near-linear time, and raises the question of whether \emph{any} algorithm can achieve this goal.
\begin{question}
Can an approximation of $W_2$ be computed in time $\Otilde(n)$?
\end{question}
In this work, we show that the answer is yes.
We employ a radically different approach from prior work in the geometric algorithms community.
Our technique, based on entropic regularization, bypasses embedding and sketching and leverages instead a \emph{low-rank} approximation to the optimal coupling.

\subsection{Problem statement}\label{subsec:problem-statement}
Let $\cX := \{x_1, \dots, x_n\} \subseteq \RR^d$, and let $\p$ and $\q$ be two distributions supported on $\cX$, given as two vectors in the simplex $\Delta_n \defeq \{\lambda \in \RR^n: \lambda \geq 0, \sum_{i=1}^n \lambda_i = 1\}$.
We identify the set $\Coup$ of couplings with nonnegative $n \times n$ matrices whose rows and columns sum to $\p$ and $\q$, respectively. This set $\Coup$ is often called the transportation polytope. We assume for normalization purposes throughout that $\cX$ lies in the Euclidean ball of radius $1$ centered at the origin; since the diameter of $\cX$ can easily be estimated within a constant factor in $O(n)$ time, this can always be achieved by translating and rescaling.

The primary goal is to estimate the cost of an optimal matching with respect to the quadratic Euclidean cost.

% \note{FB: may be add a discussion, here or somewhere, why we consider additive approximation guarantees rather than multiplicative?}

\begin{problem}
Given $\cX$, $\p$, $\q$, and $\ep \in (0, 1)$, compute $\What$ satisfying
\begin{equation*}\label{prob:val}
|\What - W^2_2(\p, \q)| \leq \ep\,.
\end{equation*}
\end{problem}
Note that the elementary inequality $(a - b)^2 \leq |a^2 - b^2|$ for $a, b \geq 0$ implies that~$\sqrt{\What}$ provides a~$\sqrt \ep$ approximation to $W_2(\p, \q)$.
This additive guarantee implies a multiplicative guarantee if $W_2(\p, \q) = \Omega(1)$; however, we do not obtain a multiplicative guarantee when $W_2(\p, \q)$ is very small. We discuss prospects for obtaining a multiplicative guarantee in Section~\ref{sec:conclusion}.
%\note{Should we comment here that we are only additive to deflect reviewer anger?}

We also consider the stronger goal of producing a near-optimal feasible coupling between $\p$ and $\q$.
\begin{problem}\label{prob:coup}
Given $\cX$, $\p$, $\q$, and $\ep \in (0, 1)$, compute $\Phat \in \Coup$ satisfying
\begin{equation*}
\sum_{i, j = 1}^n \Phat_{ij} \|x_i - x_j\|^2 \leq W^2_2(\p, \q) + \ep\,.
\end{equation*}
\end{problem}
A solution $\hat P$ to Problem~\ref{prob:coup} clearly yields a solution to Problem~\ref{prob:val}, as long as the cost $\What = \sum_{i, j = 1}^n \Phat_{ij} \|x_i - x_j\|^2$ can be computed quickly.
Note that since we are interested in algorithms with $o(n^2)$ running time, we will not be able to represent a solution $\hat P$ to Problem~\ref{prob:coup} explicitly, so we will focus on returning such a matrix in \emph{factored} form, with rank $r = o(n)$.

\subsection{Our results}
Our main result breaks the $\Otilde(n^{3/2})$ barrier for approximating the quadratic transportation cost and shows that Problems~\ref{prob:val} and~\ref{prob:coup} can be solved in near-linear time. 

\begin{theorem}[Informal, constants suppressed]
There exists a universal constant $C > 0$ and an algorithm that, given two distributions supported on $n$ points in $\RR^d$, can compute an additive $\ep$-approximation to the quadratic transportation cost in time
\begin{equation*}
\Otilde\left(\frac{n}{\eps^3} \left(\frac{C \log n}{\ep}\right)^{d}\right)\,.
%\Otilde\Big(n \Big(\frac{e \log n}{\ep}\Big)^{d+2}\Big)\,.
\end{equation*}
Moreover, the algorithm also computes a feasible coupling (in factored form) achieving this approximation.
\end{theorem}
The explicit version appears as Theorem~\ref{thm:main}, below.
%The formal version with constants is given in Theorem~\ref{thm:main}.

\subsection{Our techniques}
Our algorithm is simple.
We leverage the technique of \emph{entropic regularization}, which has been used in the machine learning literature to approximate the optimal transportation between two distributions with arbitrary costs~\citep{Cut13}, and has yielded algorithms which are fast both in theory and in practice~\citep{AltWeeRig17,GenCutPey16,DvuGasKro18,SolGoePey15,PeyCut17}.

The entropic regularization approach is based on solving the program
\begin{equation}\label{eq:regularized}
\min_{P \in \Coup} \langle C, P \rangle - \eta^{-1} H(P)\,,
\end{equation}
for a carefully chosen regularization parameter $\eta$, where $C_{ij} := \|x_i - x_j\|^2$ for $i, j \in [n]$ and $H(P) := \sum_{ij} P_{ij} \log \frac{1}{P_{ij}}$ is the entrywise Shannon entropy of the matrix $P$.
This approach was popularized by~\citet{Cut13}, though similar ideas date back to~\citet{Wil69}.

The benefit of this technique is that the solution to the regularized program~\eqref{eq:regularized} can be characterized explicitly.
\begin{fact}[{\citet[Lemma~2]{Cut13}}]\label{fact:kkt}
The minimizer of~\eqref{eq:regularized} is the \emph{unique} matrix in $\cM(\p, \q)$ of the form $D_1 K D_2$ for positive matrices $D_1$ and $D_2$, where $K_{ij} := e^{- \eta C_{ij}}$ for~$i, j \in [n]$.
\end{fact}

Finding matrices $D_1$ and $D_2$ for which $D_1 K D_2 \in \Coup$ is an instance of what is known as the \emph{matrix scaling problem}, which has a long history in the optimization and computer science literature~\citep{LinSamWig00,Gur03,RotSch89}. Several recent works have shown that the matrix scaling problem can be solved in $\Otilde(n^2)$ time by second-order methods~\citep{CohMpoTsi17,AllLiOli17}, and this quadratic dependence on $n$ is not improvable for general matrices $K$.

In this work, we adopt an older and simpler approach to the matrix scaling problem, known as Sinkhorn scaling (after~\citet{Sin67}) or the RAS method~\citep{Bac65}. Sinkhorn scaling is a simple alternating minimization scheme, which alternates between renormalizing the rows and columns of $K$ so that they have the desired marginals $\p$ and $\q$, respectively. For completeness, pseudocode for this procedure is provided in Appendix~\ref{app:sink}. This algorithm iteratively builds positive diagonal matrices $D_1$ and $D_2$ such that $D_1 K D_2$ converges to an element of $\Coup$. This method has also been shown to solve the scaling problem in time $\Otilde(n^2)$, albeit with polynomial (rather than logarithmic) dependence on the desired precision:

\begin{fact}[{\citet[Theorem~2]{AltWeeRig17}, \citet[Theorem~1]{DvuGasKro18}}]\label{fact:sink_run}
Given a positive matrix $K \in \RR^{n \times n}$, Sinkhorn scaling (Algorithm~\ref{alg:sinkhorn}) computes diagonal matrices $\tilde D_1$ and $\tilde D_2$ such that $\tilde P \defeq \tilde D_1 K \tilde D_2$ satisfies $\|\tilde P \bone - \p\|_1 + \|\tilde P^\top \bone - \q\|_1 \leq \delta$ in $O(\delta^{-1} \log \frac{n}{\delta \min_{ij} K_{ij}})$ iterations.
Moreover, the entries of $\tilde D_1$ and $\tilde D_2$ are polynomially bounded in $\delta$, $\min_{ij} K_{ij}$, and $n$ throughout the execution of the algorithm.
\end{fact}

The benefit of using this algorithm rather than a second-order method is that Sinkhorn scaling can be implemented such that each iteration requires $O(1)$ matrix-vector products with $K$. The simple but key observation is that although computing these products takes $\Theta(n^2)$ time for general positive matrices, if $K$ can be written in factored form as $V^\top V$ where $V \in \RR^{r \times n}$ for $r = o(n)$, then each iteration takes time $O(rn) = o(n^2)$.
To exploit this property, we rely on the fact that when $C_{ij} = \|x_i - x_j\|^2$, the matrix $K$ is a \emph{Gaussian kernel matrix}, with entries of the form $e^{-\eta \|x_i - x_j\|^2}$. We then appeal to the fact that Gaussian kernel matrices can be efficiently approximated by matrices of very low rank, with $r = \Otilde(1)$.

The main technical step is to show that Sinkhorn scaling can be used on this approximate matrix to produce a suitably good estimate of the optimal coupling. We then round the matrix obtained from approximate Sinkhorn scaling to the transport polytope and return the resulting matrix $\hat P$ and its cost $\hat W = \sum_{ij} \hat P_{ij} \|x_i - x_j\|^2$. Once we establish that each step can be implemented in time $O(nr) = \Otilde(n)$, the proof is complete.

\subsection{Related work}\label{sec:intro:rel}
Approximating the Wasserstein distance between discrete distributions is a fundamental problem in optimization. In particular, since it is a special case of the transportation problem, it has been the focus of a great deal of work in the combinatorial optimization community. (See~\citep[Chapter~21]{Sch03} and references therein.)
General optimal transportation problems are linear programs and can be solved in time $\Otilde(n^{2.5})$ by interior point techniques~\citep{LeeSid14}, or in time $\Otilde(n^3)$ by simple combinatorial methods~\citep{Orl93}.
In general, no algorithm for the optimal transport problem can run in time $o(n^2)$ without additional structural assumptions, since the matrix $(C_{ij})$ encoding the costs between locations has size $\Theta(n^2)$. \citet{AltWeeRig17} showed that this goal is (nearly) achievable by exhibiting an algorithm that obtained an additive $\ep$ approximation to the optimal transport cost in time $\Otilde(n^2 \ep^{-3})$, which has been improved to $\Otilde(n^2 \ep^{-1})$ in subsequent work~\citep{DvuGasKro18,quanrud,BlaJamCarSid18}. These results hold for any nonnegative cost matrix $C$ satisfying $\max_{ij} C_{ij} = O(1)$.

If the cost $C$ is metric, it is possible to obtain near-linear time algorithms. When $C_{ij}$ is an $\ell_p$ norm on $\RR^d$---which includes the $W_1$ case---\citet{ShaAga12} show that a $(1+\ep)$ multiplicative approximation can be obtained in $O(n \, \poly(\log n, \ep^{-1}))$ time in the special case where $\p$ and $\q$ are uniform distributions on $n$ points. \citet{AndNikOna14} gives an algorithm that can obtain an estimate of the $W_1$ distance between any two distributions in $\RR^2$ in time $O(n^{1 + o_\ep(1)})$, though this algorithm does not return a coupling. Both algorithms incur an exponential dependence on $d$ in the high-dimensional setting.

For the quadratic Euclidean cost---the $W_2$ case---no near-linear time algorithms are currently known. The first subquadratic algorithm was given by~\citet{PhiAga06}, who again consider the special case where $\p$ and $\q$ are uniform distributions on $n$ points in $\RR^2$. They obtain a $(1+\ep)$ multiplicative approximation in time $\Otilde(n^{3/2} \ep^{-3/2})$. This was extended to $\RR^d$ by~\citet{AgaSha14} (though still under the condition that $\p$ and $\q$ are uniform distributions), who obtain a $(1+\ep)$ multiplicative approximation in $\Otilde(n^{3/2} \ep^{-1} \tau(n, \ep))$ time, where $\tau(n, \ep)$ is the time required to query an $O(\ep)$-approximate nearest neighbor data structure for $\ell_2$. Designing such data structures in high dimensions is a delicate matter~\citep{AndIndRaz18}, but if we require that query time be polylogarithmic in $n$ and that the data structure take $n^{O(1)}$ time to build, the best results incur dependence of order $O(d/\ep)^d$ on the dimension~\citep{AryMouNet98}.
No algorithms running in time $o(n^{3/2})$ are known, even for the case of uniform distributions in $\RR^2$.

The idea of low-rank approximation to kernel matrices, which we exploit here, is common in machine learning~\citep{BacJor03,SchSmo01,FinSch01,RahRec07}. The use of low-dimensional couplings in optimal transport has been proposed for statistical purposes~\citep{SpaBigMar17,ForHutNit18}, but, to our knowledge, has not been explored from an optimization perspective.

\paragraph*{Acknowledgements.} We thank Pablo Parrilo and Philippe Rigollet for helpful discussions.

% !TEX root = ../stoc.tex
\section{Preliminaries}\label{sec:prelim}

First, some notation. 
%Throughout the paper, all matrix exponentials $e^A$ and matrix logarithms $\log A$ are entrywise.
%The entropy $H(P) := -\sum_{ij} P_{ij} \log P_{ij}$ of a matrix is also taken entrywise. 
We use the matrix norms $\|A\|_{\infty} := \max_{ij} |A_{ij}|$ and $\|A\|_1 := \sum_{ij} |A_{ij}|$. For vectors, the notation $\|\cdot\|$ refers to the $\ell_2$ norm. We write $\ball$ to denote the unit $\ell_2$ ball around the origin in dimension $d$, $\bone$ to denote the all-ones vector of dimension $n$, and $[k]$ to denote $\{1, \dots, k\}$ for positive integers $k$. The Frobenius inner product between matrices is denoted by $\langle A, B \rangle$. The notation $\Otilde(\cdot)$ hides factors of the form $\log^{O(1)} n \log^{O(1)}(1/\ep)$.

\subsection{Low-rank approximation of Gaussian kernel matrices}\label{subsec:prelim:taylor}

For the quadratic transport cost, the cost matrix $C_{ij} := \|x_i - x_j\|^2$ corresponds to the squared Euclidean distance, and thus the entrywise-exponentiated matrix $K_{ij} := \exp(-\eta C_{ij})$ in Fact~\ref{fact:kkt} is a Gaussian kernel matrix (where $\eta = \tfrac{1}{2\sigma^2}$).

\begin{defin}
The \emph{Gaussian kernel matrix} corresponding to points $x_1, \dots, x_n \in \RR^d$ and bandwith parameter $\sigma > 0$ is the matrix $K := K(x_1,\dots,x_n;\sigma) \in \RR^{n \times n}$ with entries $K_{ij} := \exp(-\tfrac{\|x_i - x_j\|^2}{2\sigma^2})$.
\end{defin}

Details about Gaussian kernel matrices can be found in, e.g.,~\citep{SteChr08, cotter2011explicit}.
Of critical importance to our approach is the fact that the spectrum of Gaussian kernel matrices decays exponentially fast, since this enables the approximation of Gaussian kernel matrices by low-rank matrices. We will use throughout the paper an explicit low-rank approximation obtained by truncating the Taylor expansion of the Gaussian kernel.

\begin{lemma}[{\citet{cotter2011explicit}}]\label{lem:taylor-gaussian}
There is a procedure $\TaylorGKM$ that, given any $x_1, \dots, x_n \in \ball$, $\sigma > 0$, and $M \in \mathbb{N}$, takes time $O(n\binom{M-1+d}{M-1})$ to output a matrix $V_M \in \RR^{\binom{M-1+d}{M-1} \times n}$ satisfying
\begin{align*}
\left\|K(x_1,\dots,x_n;\sigma) - V_M^TV_M\right\|_{\infty}
\leq
\frac{1}{M!\, \sigma^{2M}}.
\end{align*}
%Moreover, $V_M(x_1, \dots, x_n;\sigma)$ can be formed in time $O(n\binom{M-1+d}{M-1})$.
\end{lemma}

We write $V_M \gets \TaylorGKM(x_1,\dots,x_n;\sigma,M)$ to denote the procedure generating this low-rank factorization; for completeness, this is briefly described in Appendix~\ref{app:taylor}.
%An explicit such low-rank approximation is achieved by truncating the Taylor expansion of the Gaussian kernel.
We will choose $M = O(\ep^{-1} \log n)$, so that the approximation $V_M^TV_M$ has rank $O(\ep^{-1} \log n)^d$.

\subsection{Rounding to transport polytope}
It will be helpful to record the following simple guarantee about the rounding algorithm $\textsc{Round}$ due to~\citet{AltWeeRig17}. The performance guarantee is a slight variation of~\citep[Lemma~7]{AltWeeRig17}, and the runtime guarantee is immediate by definition of the algorithm. For completeness, Appendix~\ref{app:round} includes a short proof of this as well as pseudocode for the $\textsc{Round}$ procedure.
%algorithm in this paper's notation. 

\begin{lemma}\label{lem:round-alg}
Given $\p,\q \in \Delta_n$ and $F \in \RR_{\geq 0}^{n \times n}$ satisfying $\|F\|_1 = 1$, $\textsc{Round}(F, \p, \q)$ outputs a matrix $G \in \Coup$ of the form $G = D_1 F D_2 + uv^\top$ for positive diagonal matrices $D_1$ and $D_2$ satisfying
\[
\|G - F\|_1 \leq \|F\bone - \p\|_1 + \|F^T\bone - \q\|_1.
\]
Moreover, the algorithm only uses $O(1)$ matrix-vector products with $F$ and $O(n)$ additional processing time.
\end{lemma}
In particular, if $F$ is given explicitly in factored form, then the output of $\textsc{Round}(F, \p, \q)$ also has an explicit factorization with at most one additional rank $1$ term.

\subsection{Computing the quadratic transportation cost of low-rank matrices}
Given a matrix $P \in \Coup$, computing the cost $\sum_{i, j = 1}^n P_{ij} \|x_i - x_j\|^2$ na\"ively takes $\Omega(n^2)$ time. However, if $P$ is given explicitly in low-rank form, this cost can be computed more quickly.

\begin{fact}
The matrix $C$ given by $C_{ij} = \|x_i - x_j\|^2$ satisfies
\begin{equation*}
C = y \bone^\top + \bone y^\top - 2 X^\top X\,,
\end{equation*}
where $X := [x_1 | \dots | x_n] \in \RR^{d \times n}$ and $y := [\|x_1\|^2, \dots, \|x_n\|^2]^\top \in \RR^n$.
\end{fact}

\begin{lemma}\label{lem:quad-cost}
If $P \in \Coup$ is given as $\sum_{k=1}^r v_k w_k^\top$, then $\sum_{i, j = 1}^n P_{ij} \|x_i - x_j\|^2$ can be computed in $O(ndr)$ time.
\end{lemma}
\begin{proof}
The matrix $X$ and the vector $y$ can be computed in $O(nd)$ time.
Clearly $\langle P, y \bone^\top + \bone y^\top \rangle = y^\top (\p + \q)$ can be computed in $O(n)$ time, and $\langle P, X^\top X \rangle = \sum_{k=1}^r w_k X^\top X v_k$ can be computed in $O(ndr)$ time.
Therefore computing $\sum_{i, j = 1}^n P_{ij} \|x_i - x_j\|^2 = \langle P, C \rangle$ takes $O(ndr)$ time.
\end{proof}

\section{Approximating $W_2$ in $\tilde{O}(n)$ time}\label{sec:alg}

\begin{algorithm}[H]
\begin{algorithmic}[1]
\Require{$X := \{x_1, \dots, x_n\} \subseteq \RR^{d}$, $\p, \q \in \Delta_n$, $\eps > 0$}
\Ensure{$\hat P \in \Coup$, $\hat W \in \RR$}
\State $\eta \gets \tfrac{20 \log n}{\eps}$, $M \gets \tfrac{300 \log n}{\eps}$
\State $V_M \leftarrow \TaylorGKM(x_1,\dots,x_n;\tfrac{1}{\sqrt{2\eta}},M)$ \label{line:taylor} \Comment Compute low-rank approximation
\State $(D_1, D_2) \gets \textsc{Sinkhorn}(V_M^\top V_M, \p, \q, \tfrac{\eps}{10})$ \label{line:sinkhorn} \Comment Approximate Sinkhorn projection
\State $\Phat \gets \textsc{Round}(D_1 V_M^\top V_M D_2, \p, \q)$%, in the form $\Phat = D_1'V_M^TV_M D_2' + vw^T$ 
\Comment Round to feasible set
\label{line:round}
%\State $\eps \gets \min(\eps,5)$, $\eta \gets \tfrac{20 \log n}{\eps}$, $M \gets \tfrac{300 \log n}{\eps}$
%\State $\tilde K \gets V_M(x_1,\dots,x_n;\tfrac{1}{\sqrt{2\eta}})^TV_M(x_1,\dots,x_n;\tfrac{1}{\sqrt{2\eta}})$ \label{line:taylor} \Comment Compute low-rank approximation
%\State $(D_1, D_2) \gets \textsc{Sinkhorn}(\tilde K, \p, \q, \tfrac{\eps}{10})$ \label{line:sinkhorn} \Comment Approximate Sinkhorn projection
%\State $\Phat \gets \textsc{Round}(D_1 \tilde K D_2, \p, \q)$, in the form $\Phat = D_1'\tilde K D_2' + vw^T$ \Comment Round to feasible set
%\label{line:round}
\State $\hat W \gets \sum_{i,j=1}^n \hat P_{ij} \|x_i - x_j\|^2$ \Comment Compute objective value
\label{line:obj}
\State \Return{$\Phat$,  $\hat W$}
\end{algorithmic}
\caption{$\tilde{O}(n)$ time algorithm for approximating $2$-Wasserstein distance.}
\label{alg:main}
\end{algorithm}

Pseudocode for our proposed algorithm is given in Algorithm~\ref{alg:main}. It returns a feasible matrix $\hat P \in \Coup$ as well as its objective value $\hat W$.
%\par A few remarks. First, the subroutines in lines~\ref{line:taylor},~\ref{line:sinkhorn},~\ref{line:round}, and~\ref{line:obj} are described in Subsections~\ref{subsec:prelim:taylor},~\ref{subse
We emphasize that we never explicitly manipulate $n \times n$ matrices in Algorithm~\ref{alg:main}, since even writing such a matrix would require~$\Omega(n^2)$ time. Instead, the algorithm uses factored matrices (of rank at most~$1 + \binom{M-1+d}{M-1})$ throughout, which allows lines~\ref{line:taylor}--\ref{line:obj} to be computed in $o(n^2)$ time. The coupling $\hat P$ is also returned in factored form, which may enable scalable implementation of algorithms that use~$\hat P$ later in the computational pipeline. Moreover, since each step requires computing matrix-vector products, Algorithm~\ref{alg:main} is easily parallelizable.
Implementation details are deferred to Subsection~\ref{subsec:main:proof}.
%We emphasize that the matrix $\hat P$ is returned in factored form (of rank at most $1 + \binom{M-1+d}{M-1}$), since otherwise even writing this $n \times n$ matrix down would require $\Omega(n^2)$ time. This factored form may enable scalable implementations of algorithms that use $\hat P$ later in the computational pipeline. For instance, we show how to leverage this low-rank factored form to show that steps~\ref{line:taylor},~\ref{line:sinkhorn},~\ref{line:round}, and~\ref{line:obj} are all implementable in $o(n^2)$ time.

The following theorem---the main result of the paper---shows that $\hat W$ is an additive $\eps$ approximation of the true $2$-Wasserstein distance between $\p$ and $\q$.

\begin{theorem}\label{thm:main}
For any $\{x_1, \dots, x_n\} \subseteq \ball$, any $\p,\q \in \Delta_n$, and any $\eps \in (0, 1)$, Algorithm~\ref{alg:main} returns a feasible matrix $\Phat \in \Coup$ (in factored form) and scalar $\hat W = \sum_{i,j=1}^n \hat P_{ij} \|x_i - x_j\|^2$ such that $\hat W \leq W_2^2(\p, \q) + \eps$ in time
%There is an implementation of Algorithm~\ref{alg:main} such that  the algorithm terminates in
\begin{equation*}
O\left(n \frac{d \log n}{\ep^2} \left(e + \frac{900\log n}{\eps d}\right)^{d+1}\right)\,.
\end{equation*}
%$O(n\tfrac{\log n}{\eps}(d+\tfrac{1}{\eps}\log(\tfrac{n}{\eps}))e^d \max(1, O(\tfrac{\log n}{d\eps})^d ))$ time.\note{TODO: update running time.} 
%\footnote{The notation $\tilde O$ hides factors of the form $(\log n)^\alpha(\log(1/\delta))^\beta(\log(1/\ep))^\gamma$ for constants $\alpha, \beta, \gamma$. We employ the subscripts to stress that the constant depends on the quantities $d, \eta, R$.}
\end{theorem}

%By running Algorithm~\ref{alg:main} with $\eps^2$ accuracy and taking the square root of the outputted $\hat{W}$ (see discussion in Subsection~\ref{subsec:problem-statement} for details), this immediately yields an algorithm 
%for approximating $W_2(\p, \q)$ to additive $\eps$ accuracy in time $O(n\tfrac{\log n}{\eps^2}(d+\tfrac{1}{\eps^2}\log(\tfrac{n}{\eps^2}))e^d \max(1, O(\tfrac{\log n}{d\eps^2})^d ))$.\note{TODO: delete?}

The proof of Theorem~\ref{thm:main} requires a simple lemma showing that the low-rank approximation computed in line~\ref{line:taylor} is a sufficiently good approximation to the Gaussian kernel matrix.
For the rest of this section, $x_1, \dots, x_n \in \ball$ are fixed, $C \in \RR^{n \times n}$ denotes the cost matrix with entries $C_{ij} := \|x_i - x_j\|^2$, and $K \in \RR^{n \times n}$ denotes the kernel matrix $K_{ij} := \exp(- \eta C_{ij})$. We write $\tilde{K} := V_M^\top V_M$ and define $\tilde C_{ij} := \eta^{-1} \log \tilde K_{ij}$.

\begin{lemma}\label{lem:logK-to-K}
For all $i, j \in [n]$, the matrices $\tilde K$ and $\tilde C$ satisfy $\tilde{K}_{ij} \geq \half e^{-\eta}$ and
$|C_{ij} - \tilde{C}_{ij}| \leq \frac{\eps}{10}$.
%\[
%\|\log K - \log \tilde{K} \|_{\infty} \leq 1.
%\]
\end{lemma}
\begin{proof}
By Lemma~\ref{lem:taylor-gaussian}, Stirling's inequality, and our choice of $M = \tfrac{300 \log n}{\eps} \geq 2e^2 \eta$, we have $
\|K - \tilde{K}\|_{\infty}
\leq \tfrac{(2\eta)^M}{M!}
\leq  \tfrac{1}{\sqrt{2\pi M}} (\frac{2\eta e}{M})^M
\leq \half e^{-2e^2\eta}
\leq \half e^{-\eta}$.
Since the smallest element of $K$ has size at least $e^{-\eta}$, we obtain $\Ktilde_{ij} %\geq K_{ij} - |K_{ij} - \Ktilde_{ij}| 
\geq \half e^{-\eta}$,
which implies in particular that $\tilde K$ is strictly positive and that $\tilde C$ is well defined. 

Next, the bound
 $|\log K_{ij} - \log \tilde K_{ij}|
\leq \frac{|K_{ij} - \tilde K_{ij}|}{\min(K_{ij}, \tilde{K}_{ij})}
\leq \frac{\half e^{-\eta} }{\half e^{-\eta}}
= 1$ implies $|C_{ij} - \tilde{C}_{ij}| = \eta^{-1} |\log K_{ij} - \log\tilde{K}_{ij}| \leq \eta^{-1} = \tfrac{\eps}{20 \log n} \leq \tfrac{\eps}{10}$.
\end{proof}

We will also use the following standard bound on the entropy of a discrete distribution.
%several times to show that for a small enough regularization parameter, entropic regularization does not change the value of the optimal transport problem much.

\begin{fact}\label{fact:entropy}~\citep[Theorem 2.6.4]{coverthomas}
Let $P \in \RR_{\geq 0}^{n \times n}$ such that $\sum_{ij}P_{ij} = 1$. Then $H(P) \in [0,2 \log n]$.
\end{fact}

\subsection{Proof of Theorem~\ref{thm:main}}\label{subsec:main:proof}

\paragraph*{Approximation guarantee.}
%\par Let $\tilde{K} := V_M^TV_M$ and let $\tilde{C} \in \RR^{n \times n}$ be the matrix with entries $\tilde{C}_{ij} := - \eta^{-1} \log \tilde{K}_{ij}$; this is well-defined since $\tilde{K}$ has strictly positive entries by Lemma~\ref{lem:logK-to-K}.
Let $P^* \in \argmin_{P \in \Coup} \langle P, C \rangle$ be any optimal solution for the original problem, 
%$P^{\eta} := \argmin_{P \in \Coup} \langle P, C \rangle - \eta^{-1} H(P)$ be the (unique) optimal solution to the regularized problem,
$\tilde{P} := \argmin_{P \in \Coup} \langle P, \tilde{C} \rangle - \eta^{-1} H(P)$ be the (unique) optimal solution to the regularized problem with the cost matrix $\tilde{C}$, and $P' := D_1\tilde{K}D_2$ be the approximately scaled matrix obtained in line~\ref{line:sinkhorn}.
%$\tilde{P} = \Sinkcoup(e^{-\eta \tilde{C}})$
We bound the suboptimality gap $\langle \hat{P}, C \rangle - \langle P^*, C \rangle$ by decomposing it as:
%\begin{align*}
%\underbrace{\langle P^{\eta} - P^*, C \rangle}_\textbf{(i)}
%+
%\underbrace{\langle \tilde{P}, \tilde{C} \rangle - \langle \Peta, C \rangle}_\textbf{(ii)}
%+
%\underbrace{\langle P' - \tilde{P}, \tilde{C} \rangle}_\textbf{(iii)}
%+
%\underbrace{\langle \hat{P} - P', \tilde{C} \rangle}_\textbf{(iv)}
%+
%\underbrace{\langle \hat{P}, C - \tilde{C} \rangle}_\textbf{(v)}
%\end{align*}
\begin{align*}
%\underbrace{\langle P^{\eta} - P^*, C \rangle}_\textbf{(i)}
%+
\underbrace{\langle \tilde{P}, \tilde{C} \rangle - \langle P^*, C \rangle}_\textbf{(i)}
+
\underbrace{\langle P' - \tilde{P}, \tilde{C} \rangle}_\textbf{(ii)}
+
\underbrace{\langle \hat{P} - P', \tilde{C} \rangle}_\textbf{(iii)}
+
\underbrace{\langle \hat{P}, C - \tilde{C} \rangle}_\textbf{(iv)}
\end{align*}

%It suffices to show that each of these terms is bounded above by $\tfrac{\eps}{5}$.
We now bound each of these terms individually so that their sum is bounded above by $\eps$.
\begin{enumerate}
%\item[(i)] By definition of $P^{\eta}$, $\langle  P^{\eta}, C \rangle - \eta^{-1}H(P^{\eta}) \leq \langle P^*, C \rangle - \eta^{-1}H(P^*)$. Using Fact~\ref{fact:entropy}, we conclude that $\langle P^{\eta} - P^*, C \rangle \leq \eta^{-1}(H(P^{\eta}) - H(P^*)) \leq 2 \eta^{-1}\log n = \tfrac{\eps}{10}$.
\item[(i)] By definition of $\tilde{P}$, $\langle \tilde{P}, \tilde{C} \rangle - \eta^{-1}H(\tilde{P}) \leq \langle P^{*}, \tilde{C} \rangle - \eta^{-1}H(P^{*})$. Thus $\langle \tilde{P}, \tilde{C} \rangle - \langle P^*, C \rangle  \leq \eta^{-1} (H(\tilde{P}) - H(P^*)) + \langle  P^{*}, \tilde{C} - C \rangle \leq 2\eta^{-1} \log n + \|\tilde{C} - C\|_{\infty} \leq \tfrac{\eps}{5}$ where we have used Fact~\ref{fact:entropy}, H\"older's inequality, and Lemma~\ref{lem:logK-to-K}.
%\item[(ii)] By Fact~\ref{fact:kkt}, $P^{\eta} = \Sinkcoup(e^{-\eta C})$ and $\tilde{P} = \Sinkcoup(e^{-\eta \tilde{C}})$. Thus by applications of H\"older's inequality, Proposition~\ref{prop:sink-stable}, and then Lemma~\ref{lem:logK-to-K}, we conclude that
%$
%\langle C, \tilde{P} - P^{\eta} \rangle
%\leq \|C\|_{\infty} \|\tilde{P} - P^{\eta}\|_1
%\leq \|C\|_{\infty} \|\log K - \log \tilde{K}\|_{\infty}
%\leq \tfrac{\eps}{5}$.
%\item[(iii)] Again using H\"older's inequality and Lemma~\ref{lem:logK-to-K}, 
%$\langle \tilde{C} - C, \tilde{P}\rangle \leq \|\tilde{C} - C\|_{\infty} \|\tilde{P}\|_1 = \eta^{-1}\|\log K - \log \tilde{K}\|_{\infty} \leq \tfrac{\eps}{5}$.
\item[(ii)] Let $\p' := P'\bone$ and let $\q' := (P')^\top \bone$. By Lemma~\ref{lem:round-alg} there exists a matrix $G \in \Coupp$ such that $\|G - \tilde{P}\|_1 \leq \|\p-\p'\|_1 + \|\q - \q'\|_1$. Now by Fact~\ref{fact:kkt}, $P' = \argmin_{P \in \Coupp} \langle P, \tilde{C} \rangle - \eta^{-1}H(P)$, thus $\langle P', \tilde{C} \rangle - \eta^{-1}H(P') \leq \langle G, \tilde{C}\rangle - \eta^{-1}H(G)$. Rearranging yields $\langle P' - \tilde{P}, \tilde{C} \rangle \leq \eta^{-1}(H(P') - H(G)) + \langle G - \tilde{P}, \tilde{C} \rangle$. The first term is bounded above by $\tfrac{\eps}{10}$ by Fact~\ref{fact:entropy}, and the second is bounded above by
\begin{align}
\|\tilde{C}\|_{\infty} (\|\p - \p'\|_1 + \|\q - \q'\|_1) \leq \tfrac{\eps}{5},
\label{eq:proof-main:nasty}
\end{align}
since $\|\p - \p'\|_1 + \|\q - \q'\|_1 \leq \tfrac{\eps}{10}$ by line~\ref{line:sinkhorn}, and since $\|\tilde{C}\|_{\infty} \leq \|C\|_{\infty} + \|C - \tilde{C}\|_{\infty} \leq 2$ by Lemma~\ref{lem:logK-to-K}. We conclude that term (ii) is bounded above by $\tfrac{3\eps}{10}$.
\item[(iii)] By H\"older's inequality and then Lemma~\ref{lem:round-alg}, $\langle \hat{P} - D_1\tilde{K}D_2, \tilde{C} \rangle \leq \|\tilde{C}\|_{\infty} \|\hat{P} - D_1\tilde{K}D_2\|_1 \leq  \|\tilde{C}\|_{\infty} (\|\p - \p'\|_1 + \|\q - \q'\|_1)$. This is bounded above by $\tfrac{\eps}{5}$ by~\eqref{eq:proof-main:nasty}.
\item[(iv)] By H\"older's inequality and Lemma~\ref{lem:logK-to-K}, $\langle \hat{P}, C - \tilde{C} \rangle \leq \|\hat{P}\|_1 \|C - \tilde{C}\|_{\infty} \leq \frac{\eps}{10}$.
\end{enumerate}

\paragraph*{Runtime guarantee.} All computations in Algorithm~\ref{alg:main} are done implicitly by only maintaining the matrices $D_1V_M^\top V_MD_2$ and $\hat{P}$ in factored form. Let $r := \binom{M-1+d}{d} \leq \big(\frac{(M+d)e}{d}\big)^d \leq (e + \frac{900 \log n}{\ep d})^d$. Line~\ref{line:taylor} takes $O(nr)$ time by Lemma~\ref{lem:taylor-gaussian}. Now since $V_M^\top V_M$ is in rank-$r$ factored form, multiplying it by a vector takes $O(nr)$ operations, and by Fact~\ref{fact:sink_run} the algorithm manipulates numbers with at most $O(\log\tfrac{1}{\eps} + \eta + \log n  ) = O(\tfrac{\log n}{\eps})$ bits. Line~\ref{line:sinkhorn} therefore takes $O(\frac{nr}{\eps^3} \log^2 n)$ time by Fact~\ref{fact:sink_run} and Lemma~\ref{lem:logK-to-K}.
Line~\ref{line:round} takes $O(nr)$ time by Lemma~\ref{lem:round-alg}, and line~\ref{line:obj} takes $O(nrd)$ time by Lemma~\ref{lem:quad-cost}.
%The $o(n^2)$ speedup in the last three lines is due to the fact that the matrices with which we multiply vectors are given in explicit low-rank factored form (all have rank $\leq r+1$).
%Finally, we show how line~\ref{line:obj} can be implemented quickly. A simple calculation gives $
%\sum_{i,j=1}^n P_{ij} \|x_i - x_j\|^2
%%= \sum_{i,j=1}^n P_{ij} \left(\|x_i\|^2 + \|x_j\|^2 - 2 x_i^T x_j \right)
%=
%- 2 \sum_{i,j=1}^n P_{ij} (X^TX)_{ij} 
%+ \sum_{i=1}^n (\p_i + \q_i )\|x_i\|^2$, where $X := [x_1 | \dots | x_n] \in \RR^{d \times n}$. The latter sum can clearly be computed in $O(nd)$ time. The former sum can be computed in $O(d(dn + nr + r^2))$ time using the simple and standard trace trick $\sum_{ij} P_{ij} (X^TX)_{ij} = \Tr(PX^TX) = \Tr(XPX^T) = \sum_{i=1}^{d} e_iXPX^Te_i$, where $e_i$ is the $i$th standard basis vector.
The entire algorithm therefore takes $O(nrd + \frac{nr}{\ep^3} \log^2 n) = O\left(nr \frac{d \log n}{\ep^2}(1 + \frac{\log n}{\ep d})\right)$ time.
%Finally, we note that throughout Algorithm~\ref{alg:main}, it suffices to use $O(\eta) = O(\eps^{-1} \log n)$ bits of precision, since this is the accuracy required by Lemma~\ref{lem:logK-to-K}.
% \note{Jon: should we elaborate briefly on (non)amplification of bit complexity by Sinkhorn, by citing the Kalantari/Khachiyan paper?}
\qed

% !TEX root = ../stoc.tex
\section{Conclusion}\label{sec:conclusion}
In this paper, we have given a simple algorithm based on entropic regularization that approximates the quadratic transport metric in near-linear time.

One interesting direction for future work is whether this result can be leveraged to obtain $(1+\ep)$ \emph{multiplicative} approximations to $W_2$ in near-linear time. We compute in this work a coupling with low-rank structure, and this approach naturally lends itself to additive guarantees. For instance, if $\p = \q$, then the unique optimal coupling has cost $0$, whereas any low-rank coupling is bound to incur nonzero cost. Obtaining a multiplicative guarantee may therefore require additional techniques.

Another interesting question is whether our techniques extend to other metrics. More generally, exploring the connection between low-rank approximation and geometric algorithms seems a promising direction for future research.

%\addcontentsline{toc}{section}{References}
\bibliographystyle{plainnat}
{\footnotesize
\bibliography{stoc}}
%\bibliography{stoc}

% !TEX root = ../stoc.tex
\appendix
\section{Deferred details about subroutines}\label{app}

\subsection{Explicit low-rank approximation of Gaussian kernel matrices}\label{app:taylor}

Here, we briefly recall an explicit low-rank approximation of Gaussian kernel matrices obtained by truncating the Taylor series of the Gaussian kernel. For details, see~\citep{cotter2011explicit}.
\par First consider two vectors $x, y \in \RR^d$. Take a truncation of the Taylor expansion $\exp\left(\tfrac{\langle x, y \rangle}{\sigma^2}\right) = \sum_{m=0}^{\infty} \tfrac{1}{m!} \left(\tfrac{\langle x, y \rangle}{\sigma^2} \right)^m$, expand $\langle x, y \rangle^m$ into monomials, group like-terms, and then multiply both sides by $\exp(-\tfrac{\|x\|_2^2+\|y\|_2^2}{2\sigma^2})$ to obtain the approximation
\begin{align}
\exp\left(-\frac{\|x-y\|_2^2}{2\sigma^2}\right) \approx \sum_{\vec{v} \in \mathbb{N}^d,\, \sum_{i=1}^d \vec{v}_i \leq M } \psi_{\vec{v}}(x;\sigma)\psi_{\vec{v}}(y;\sigma), 
\label{eq:taylor}
\end{align}
where $\psi_{\vec{v}}(x;\sigma) \propto \exp(-\tfrac{\|x\|_2^2}{2\sigma^2}) \prod_{j=1}^d x_j^{\vec{v}_j}$. Concatenating these features $\{\psi_{\vec{v}}(x;\sigma) : \vec{v} \in \mathbb{N}^d,\, \sum_{i=1}^d v_i < M \}$ into a so-called ``feature vector'' $\Phi_M(x;\sigma)$, the right hand side of~\eqref{eq:taylor} can be re-written simply as $\Phi_M(x;\sigma)^T\Phi_M(y;\sigma)$. Extending to Gaussian kernel matrices is now simple: approximate $K(x_1,\dots,x_n;\sigma)$ by the Gram matrix $V_M^TV_M$ where $V_M := [\Phi_M(x_1;\sigma), \dots, \Phi_M(x_n;\sigma)]$. This subroutine for forming $V_M$ will be denoted in the paper by $\TaylorGKM(x_1, \dots, x_n; \sigma, M)$.
\par The approximation error in~\eqref{eq:taylor} is controlled by Taylor's Theorem. This immediately yields an infinity-norm guarantee on the approximation error $K \approx V_M^TV_M$. These guarantees are recorded in Lemma~\ref{lem:taylor-gaussian}; proofs can be found in~\citep{cotter2011explicit}.

\subsection{Pseudocode for Sinkhorn algorithm}\label{app:sink}
\begin{algorithm}[H]
\begin{algorithmic}[1]
\Require{$\tilde K$ (in factored form $V^\top V$), $\p, \q \in \Delta_n$, $\delta > 0$}
\Ensure{Positive diagonal matrices $D_1, D_2 \in \RR^{n \times n}$}
\State $\tau \gets \tfrac{\delta}{8}$, $D_1, D_2 \gets I_{n \times n}$, $k \gets 0$
\State $\p' \gets (1-\tau)\p + \frac{\tau}{n} \bone$, $\q' \gets (1-\tau)\q + \frac{\tau}{n} \bone$  \Comment Round $\p$ and $\q$
\While{$\|D_1 \tilde K D_2 \bone - \p'\|_1 + \|(D_1 \tilde K D_2)^\top \bone - \q'\|_1 \leq \tfrac{\delta}{2}$}
\State $k \gets k + 1$
\If{$k$ odd}
\State $(D_1)_{ii} \gets \p'_i/(D_1 \tilde K D_2 \bone)_i$ for $i = 1, \dots, n$.  \Comment Renormalize rows
\Else
\State $(D_2)_{jj} \gets \q'_j/((D_1 \tilde K D_2)^\top \bone)_j$ for $j = 1, \dots, n$. \Comment Renormalize columns
\EndIf
\EndWhile
\State \Return{$D_1$, $D_2$}
\end{algorithmic}
\caption{\textsc{Sinkhorn}}
\label{alg:sinkhorn}
\end{algorithm}

\subsection{Pseudocode for rounding algorithm}\label{app:round}
For completeness, here we briefly recall the rounding algorithm $\textsc{Round}$ from~\citep{AltWeeRig17} and give a short proof of Lemma~\ref{lem:round-alg} by adapting slightly their proof of Lemma 7.
\par It will be convenient to develop a little notation. For a vector $x \in \R^n$, $\diag(x)$ denotes the $n \times n$ diagonal matrix with diagonal entries $[\diag(x)]_{ii} = x_i$. For a matrix $A$, $r(A) := A\bone$ and $c(A) := A^T \bone$ denote the row and column marginals of $A$, respectively. We further denote $r_i(A) = [r(A)]_i$ and similarly $c_j(A) := [c(A)]_j$. 

\begin{algorithm}[H]
\begin{algorithmic}[1]
\Require{$F \in \RR^{n \times n}$ and $\p,\q \in \Delta_n$}
\Ensure{$G \in \Coup$}
\State $X \gets \diag(x)$, where $x_i := \tfrac{\p_i}{r_i(F)} \wedge 1$
\State $F' \gets XF$
\State $Y \gets \diag(y)$, where $y_j := \tfrac{\q_j}{c_j(F')} \wedge 1$
\State $F'' \gets F'Y$
\State $\errr \gets \p - r(F'')$, $\errc \gets \q - c(F'')$
\State Output $G \gets F'' + \errr \errc^T / \|\errr\|_1$
\end{algorithmic}
\caption{\textsc{Round} (from Algorithm 2 in~\citep{AltWeeRig17})}
\label{alg:round}
\end{algorithm}

\begin{proof}[Proof of Lemma~\ref{lem:round-alg}]
The runtime claim is clear. Next, let $\Delta := \|F\|_1 - \|F''\|_1 = \sum_{i=1}^n (r_i(F) - \p_i)_+ + \sum_{j=1}^n (c_j(F') - \q_j)_+$ denote the amount of mass removed from $F$ to create $F''$. Observe that $\sum_{i=1}^n (r_i(F) - \p_i)_+ = \half \|r(F) - \p\|_1$. Since $F' \leq F$ entrywise, we also have $\sum_{j=1}^n (c_j(F') - \q_j)_+ \leq \sum_{j=1}^n (c_j(F) - \q_j)_+ = \half\|c(F) - \q\|_1$. Thus $\Delta \leq \half ( \|r(F) - \p\|_1 + \|c(F) - q\|_1)$. The proof is complete since $\|F - G\|_1 
\leq
\|F - F''\|_1 + \|F'' - G\|_1
=
2\Delta$.
\end{proof}

\end{document}